\documentclass[12pt,english]{article}
\usepackage[T1]{fontenc}
\usepackage[latin9]{inputenc}
\usepackage{verbatim}
\usepackage{url}
\usepackage{amsthm}
\usepackage{amsmath}
\usepackage{amssymb}
\usepackage[all]{xy}

\makeatletter
\theoremstyle{plain}
\newtheorem{thm}{Theorem}[section]
  \theoremstyle{definition}
  \newtheorem{defn}[thm]{Definition}
  \theoremstyle{remark}
  \newtheorem{claim}[thm]{Claim}
  \theoremstyle{plain}
  \newtheorem{fact}[thm]{Fact}
  \theoremstyle{plain}
  \newtheorem{lem}[thm]{Lemma}
  \newcounter{casectr}
  \newenvironment{caseenv}
  {\begin{list}{{\itshape\ Case} \arabic{casectr}.}{%
   \setlength{\leftmargin}{\labelwidth}
   \addtolength{\leftmargin}{\parskip}
   \setlength{\itemindent}{\listparindent}
   \setlength{\itemsep}{\medskipamount}
   \setlength{\topsep}{\itemsep}}
   \setcounter{casectr}{0}
   \usecounter{casectr}}
  {\end{list}}
  \theoremstyle{remark}
  \newtheorem*{rem*}{Remark}
  \theoremstyle{plain}
  \newtheorem{cor}[thm]{Corollary}
  \theoremstyle{plain}
  \newtheorem{prop}[thm]{Proposition}

\usepackage{fullpage}

\@ifundefined{showcaptionsetup}{}{%
 \PassOptionsToPackage{caption=false}{subfig}}
\usepackage{subfig}
\makeatother

\usepackage{babel}

\begin{document}

\title{The hat problem on a directed graph}

\author{Rani Hod%
\thanks{School of Computer Science, Raymond and Beverly Sackler Faculty of
Exact Sciences, Tel Aviv University, Tel Aviv, Israel. E-mail: rani.hod@cs.tau.ac.il.
Research supported by an ERC advanced grant.%
} \and Marcin Krzywkowski%
\thanks{Faculty of Applied Physics and Mathematics, Gda\'{n}sk University
of Technology, Narutowicza 11/12, 80--952 Gda\'{n}sk, Poland. E-mail:
fevernova@wp.pl%
}}

\date{}
\maketitle
\begin{abstract}
A team of players plays the following game. After a strategy session,
each player is randomly fitted with a blue or red hat. Then, without
further communication, everybody can try to guess simultaneously his
or her own hat color by looking at the hat colors of other players.
Visibility is defined by a directed graph; that is, vertices correspond
to players, and a player can see each player to whom she or he is
connected by an arc. The team wins if at least one player guesses
his hat color correctly, and no one guesses his hat color wrong; otherwise
the team loses. The team aims to maximize the probability of a win,
and this maximum is called the hat number of the graph.

Previous works focused on the problem on complete graphs and on undirected
graphs. Some cases were solved, e.g., complete graphs of certain orders,
trees, cycles, bipartite graphs. These led Uriel Feige to conjecture
that the hat number of any graph is equal to the hat number of its
maximum clique.

We show that the conjecture does not hold for directed graphs, and
build, for any fixed clique number, a family of directed graphs of
asymptotically optimal hat number. We also determine the hat number
of tournaments to be one half.

\noindent \textbf{Keywords:} hat problem, directed graph, skeleton,
clique number.

\noindent \textbf{\AmS\; Subject Classification:} 05C20, 05C69, 91A12,
91A43. 
\end{abstract}
\global\long\def\prob#1{\mathbb{P}\left(#1\right)}
\global\long\def\ceil#1{\left\lceil #1\right\rceil }
 \global\long\def\floor#1{\left\lfloor #1\right\rfloor }
\global\long\def\skel#1{\mathrm{skel}\left(#1\right)}
%
{}

\section{\label{sec:introduction}Introduction}

\noindent In the hat problem, a team of $n$ players enters a room
and a blue or red hat is randomly and independently placed on the
head of each player. Each player can see the hats of all of the other
players but not his own. No communication of any sort is allowed,
except for an initial strategy session before the game begins. Once
they have had a chance to look at the other hats, each player must
simultaneously guess the color of his own hat or pass. The team wins
if at least one player guesses his hat color correctly and no one
guesses his hat color wrong; otherwise the team loses. The aim is
to maximize the probability of winning.

\paragraph{Origin.}

The hat problem with seven players, called the {}``seven prisoners
puzzle'', was formulated by Todd Ebert in his Ph.~D. Thesis~\cite{Ebert-phd}.
It is often posed as a puzzle (e.g., in the Berkeley Riddles~\cite{Berkeley})
and was also the subject of articles in the popular media~\cite{Blum-diezeit,Pol-abcnews,Rob-nytimes}.

The hat problem with $q\ge2$ possible colors was investigated in~\cite{LS02}.
Alon~\cite{Alon08} proved that the $q$-ary hat number of the complete
graph tends to one as the graph grows.

Many other variations of the problems exist, among them a random but
non-uniform hat color distribution~\cite{GKRT07}, an adversarial
allocation of hat from a pool known by the players~\cite{Feige04},
a variation in which passing is not allowed~\cite{BHKL08}, and many
more.

\paragraph{Our focus.}

We consider the hat problem on a graph, where vertices correspond
to players and a player can see each player to whom he is connected
by an edge. We seek to determine the hat number of the graph, that
is, the maximal chance of success for the hat problem in it. This
variation of the hat problem was first considered in~\cite{Krzy-trees}.

Note that the hat problem on the complete graph is equivalent to the
original hat problem. This case was solved for $2^{k}-1$ players
in~\cite{EMV03} and for $2^{k}$ players in~\cite{CHLL-book}.
In~\cite{LS02} it was shown that a strategy for $n$ players in
the complete graph is equivalent to a covering code of radius 1 in
the Hamming cube.

The hat problem was solved for trees~\cite{Krzy-trees}, cycles~\cite{Feige09+,Krzy-C4,Krzy-C5},
bipartite graphs~\cite{Feige09+}, perfect graphs~\cite{Feige09+},
and planar graphs containing a triangle~\cite{Feige09+}. Feige~\cite{Feige09+}
conjectured that for any graph the hat number is equal to the hat
number of its maximum clique. He proved this for graphs with clique
number $2^{k}-1$. The simplest remaining open case is thus triangle-free
graphs.

In this paper we consider the hat problem on directed graphs. Under
an appropriate definition of the clique number for directed graphs,
we construct families of digraphs with a~fixed clique number the
hat number of which is asymptotically optimal.

\section{\label{sec:preliminaries}Preliminaries}

We begin with some definitions regarding directed graphs (digraphs)
and undirected graphs.
\begin{defn}
The \emph{skeleton} of a digraph $D=\left(V,A\right)$, denoted by
$\skel D$, is the undirected graph on the vertex set $V$ in which
$x$ and $y$ are adjacent if both arcs between them belong to the
set $A$; that is, if they form a directed $2$-cycle in $D$. 
\end{defn}

\begin{defn}
The \emph{clique number} of a digraph $D$ is the clique number of
its skeleton; that is, $\omega\left(D\right)=\omega\left(\skel D\right)$.
\end{defn}

\begin{defn}
The \emph{transpose} of a digraph $D=\left(V,A\right)$ is the digraph
$D^{t}=\left(V,A^{t}\right)$, where $A^{t}=\left\{ \left(x,y\right):\left(y,x\right)\in A\right\} $. 
\end{defn}
Slightly abusing notation, we identify a digraph $D$ with its (undirected)
skeleton in the case that $D=D^{t}$; that is, if all arcs of $D$
have anti-parallel counterparts.

\bigskip{}

Fix a digraph $D=\left(V,A\right)$ on the vertex set $V=\left\{ v_{1},v_{2},\ldots,v_{n}\right\} $.
We proceed with a more precise definition of the hat problem on $D$.
\begin{defn}
A (hat) configuration is a function $c:V\to\left\{ \textrm{blue},\textrm{red}\right\} $,
assigning the hat color $c\left(v\right)$ to the vertex $v\in V$.
Naturally, there are $2^{n}$ possible configurations. 
\end{defn}

\begin{defn}
The view of a vertex $v\in V$ of a configuration $c:V\to\left\{ \textrm{blue},\textrm{red}\right\} $
is the restriction of $c$ to vertices seen by $v$, namely the function
$c^{v}=c|_{N^{+}\left(v\right)}$. Since the domain of $c^{v}$ is
$N^{+}\left(v\right)$, a set of size $d^{+}\left(v\right)$, the
number of possible views for $v$ is $2^{d^{+}\left(v\right)}$. Note
that $2^{n-d^{+}\left(v\right)}$ different configurations share any
single view of $v$.
\end{defn}
Sometimes we will regard configurations and views as binary vectors
of the respective length; that is, $c\in\left\{ \textrm{blue},\textrm{red}\right\} ^{n}$
and $c^{v}\in\left\{ \textrm{blue},\textrm{red}\right\} ^{d^{+}\left(v\right)}$
.
\begin{defn}
An individual strategy for the vertex $v\in V$ is a function mapping
views to guesses; that is, $g^{v}:\left\{ \textrm{blue},\textrm{red}\right\} ^{d^{+}\left(v\right)}\to\left\{ \textrm{blue},\textrm{red},\textrm{pass}\right\} $.
A (team) strategy is a sequence $\mathcal{S}=\left(g^{1},\ldots,g^{n}\right)$
of $n$ individual strategies, where $g^{i}$ is a strategy for $v_{i}$. 
\end{defn}

\begin{defn}
For a configuration $c\in\left\{ \textrm{blue},\textrm{red}\right\} ^{n}$
and an individual strategy $g^{v}$ for a vertex $v\in V$, we say
that $v$ guesses correctly if $g^{v}\left(c^{v}\right)=c\left(v\right)$
and guesses wrong if $g^{v}\left(c^{v}\right)\notin\left\{ \textrm{pass},c\left(v\right)\right\} $.
For a configuration $c\in\left\{ \textrm{blue},\textrm{red}\right\} ^{n}$
and a strategy $\mathcal{S}$, we say that the team wins if at least
one vertex guesses correctly and no vertex guesses wrong.
\end{defn}

\begin{defn}
The chance of success $\prob{\mathcal{S}}$ of a strategy $\mathcal{S}$
is the probability that the team wins, using $\mathcal{S}$, at a
configuration selected uniformly at random from $\left\{ \textrm{blue},\textrm{red}\right\} ^{n}$.
The hat number of the digraph $D$ is the maximum  $h\left(D\right)=\max_{\mathcal{S}}\prob{\mathcal{S}}$.
A strategy $\mathcal{S}$ is optimal for $D$ if $\prob{\mathcal{S}}=h\left(D\right)$.
\end{defn}
By solving the hat problem on a digraph $D$ we mean finding $h\left(D\right)$.

\bigskip{}

The hat problem on undirected graphs was treated in~\cite{Feige09+,Krzy-trees}.
We now cite four claims that generalize to digraphs with little or
no change.
\begin{claim}
\label{clm:subgraph}For every two digraphs $D$ and $E$ such that
$D\subseteq E$ we have $h\left(D\right)\le h\left(E\right)$.
\end{claim}

\begin{claim}
\label{clm:at-least-half}For every digraph $D$ we have $h\left(D\right)\ge1/2$.
\end{claim}

\begin{claim}
\label{clm:always-guess}Let $D$ be a digraph and let $v$ be a vertex
of $D$. If $\mathcal{S}$ is a strategy for $D$ in which $v$ always
attempts to guess its color, then $\prob{\mathcal{S}}\le1/2$.
\end{claim}

\begin{claim}
\label{clm:never-guess}Let $D$ be a digraph and let $v$ be a vertex
of $D$. If $\mathcal{S}$ is an optimal strategy for $D$ in which
$v$ never attempts to guess its color, then $h\left(D\right)=h\left(D-v\right)$.
\end{claim}
Combining Claims~\ref{clm:at-least-half}, \ref{clm:always-guess}
and \ref{clm:never-guess} we get the following.
\begin{claim}
\label{clm:drop-blind-vertices}Let $D$ be a digraph and let $v$
be a vertex of $D$. If $v$ has no outgoing arcs, i.e., $d^{+}\left(v\right)=0$,
then $h\left(D\right)=h\left(D-v\right)$.
\end{claim}
%
{}

\section{\label{sec:constructions}Constructions}

For an undirected graph $G$, it is known that if $G$ contains a
triangle, then $h\left(G\right)\ge3/4$, and it is conjectured in~\cite{Feige09+}
that if $G$ is triangle-free, then $h\left(G\right)=1/2$. Do directed
graphs introduce anything in between? The answer is yes.

\bigskip{}

Let us consider the hat problem on the digraph $D_{1}$ given in Figure~\ref{fig:5/8-example}.
\begin{fact}
$h(D_{1})=5/8$.
\end{fact}
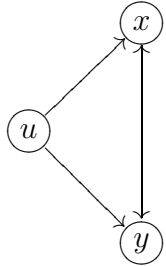
\begin{figure}[h]
\noindent \begin{centering}
\SelectTips{cm}{}
\xymatrix{
&*++[o][F]{x}\ar[dd]\\
*++[o][F]{u}\ar[dr]\ar[ur]\\
&*++[o][F]{y}\ar[uu]
}
\par\end{centering}

\caption{\label{fig:5/8-example}The directed graph $D_{1}$}

\end{figure}

We omit the proof of this fact in favor of extending $D_{1}$ to a
construction of a family $\left\{ D_{n}\right\} _{n=0}^{\infty}$
of semi-complete digraphs that asymptotically achieve hat number $2/3$,
with the property that the $\omega\left(D_{n}\right)=2$.%
\footnote{Moreover, the skeleton of $D_{n}$ is a matching of size $n$ plus
an isolated vertex. For short, we write $\skel{D_{n}}=nK_{2}\cup K_{1}$.%
}
\begin{defn}
Given two disjoint digraphs $C$ and $D$, we define the \emph{directed
union} of $C$ and $D$, denoted by $C\to D$, as the disjoint union
of these two digraphs with the additional arcs from all vertices of
$C$ to all vertices of $D$. Note that this operator is associative;
that is, $C\to\left(D\to E\right)=\left(C\to D\right)\to E$ for any
three digraphs $C$, $D$ and $E$. Thus, the notation $C\to D\to E$
is unambiguous. We denote the directed union of $n$ disjoint copies
of a digraph $D$ by $D^{\to n}=\underbrace{D\to D\to\cdots\to D}_{n}$.
\end{defn}
Expressed in the terms of directed union, $D_{1}=K_{1}\to K_{2}$.
We extend this to a family of digraphs by defining $D_{n}=K_{1}\to K_{2}^{\to n}$.
Note that the family $\left\{ D_{n}\right\} _{n=0}^{\infty}$ satisfies
the recurrence relation $D_{n+1}=D_{n}\to K_{2}$ for $n\in\mathbb{N}$.

In Figure~\ref{fig:2/3-family} we give examples of $D_{n}$ for
$n=2$, $n=3$, and a general $n$.

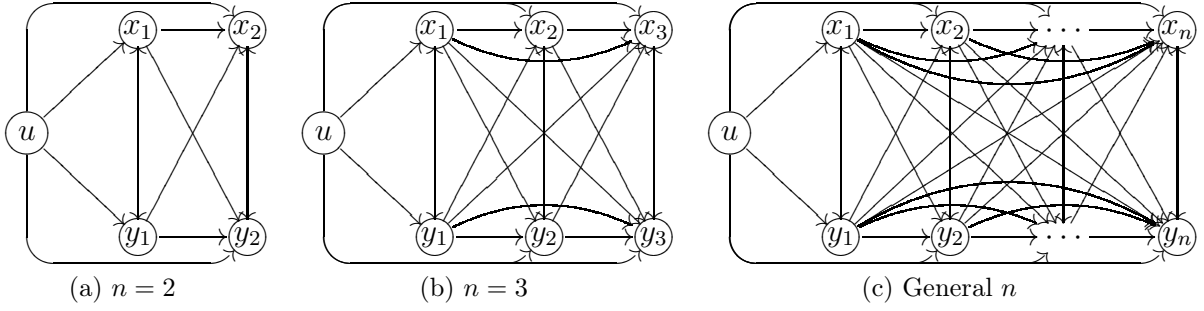
\begin{figure}[h]
\subfloat[$n=2$]{\SelectTips{cm}{}
\xymatrix{
&*+[o][F]{x_{1}}\ar[dd]\ar[r]\ar[ddr]&*+[o][F]{x_{2}}\ar[dd]\\
*++[o][F]{u}\ar[dr]\ar[ur]\ar`u_u[u]`_r[urr]`_rd[urr][urr]\ar`d_d[d]`^r[drr]`^ru[drr][drr]\\
&*+[o][F]{y_{1}}\ar[uu]\ar[r]\ar[uur]&*+[o][F]{y_{2}}\ar[uu]
}

}\subfloat[$n=3$]{\SelectTips{cm}{}
\xymatrix{
&*+[o][F]{x_{1}}\ar[dd]\ar[r]\ar[ddr]\ar@/_1pc/[rr]\ar[ddrr]
&*+[o][F]{x_{2}}\ar[dd]\ar[r]\ar[ddr]
&*+[o][F]{x_{3}}\ar[dd]\\
*++[o][F]{u}\ar[dr]\ar[ur]
\ar`u_u[u]`_r[urr]`_rd[urr][urr]
\ar`d_d[d]`^r[drr]`^ru[drr][drr]
\ar`u_u[u]`_r[urrr]`_rd[urrr][urrr]
\ar`d_d[d]`^r[drrr]`^ru[drrr][drrr]\\
&*+[o][F]{y_{1}}\ar[uu]\ar[r]\ar[uur]\ar@/^1pc/[rr]\ar[uurr]
&*+[o][F]{y_{2}}\ar[uu]\ar[r]\ar[uur]
&*+[o][F]{y_{3}}\ar[uu]
}

}\subfloat[General $n$]{\SelectTips{cm}{}
\xymatrix{
&*+[o][F]{x_{1}}\ar[dd]\ar[r]\ar[ddr]\ar@/_1pc/[rr]\ar[ddrr]\ar@/_4ex/[rrr]\ar[ddrrr]
&*+[o][F]{x_{2}}\ar[dd]\ar[r]\ar[ddr]\ar@/_1pc/[rr]\ar[ddrr]
&\cdots\ar[r]\ar[ddr]\ar[dd]
&*+[o][F]{x_{n}}\ar[dd]\\
*++[o][F]{u}\ar[dr]\ar[ur]
\ar`u_u[u]`_r[urr]`_rd[urr][urr]
\ar`d_d[d]`^r[drr]`^ru[drr][drr]
\ar`u_u[u]`_r[urrr]`_rd[urrr][urrr]
\ar`d_d[d]`^r[drrr]`^ru[drrr][drrr]
\ar`u_u[u]`_r[urrrr]`_rd[urrrr][urrrr]
\ar`d_d[d]`^r[drrrr]`^ru[drrrr][drrrr]\\
&*+[o][F]{y_{1}}\ar[uu]\ar[r]\ar[uur]\ar@/^1pc/[rr]\ar[uurr]\ar@/^4ex/[rrr]\ar[uurrr]
&*+[o][F]{y_{2}}\ar[uu]\ar[r]\ar[uur]\ar@/^1pc/[rr]\ar[uurr]
&\cdots\ar[r]\ar[uur]\ar[uu]
&*+[o][F]{y_{n}}\ar[uu]\\
}

}

\bigskip{}

\caption{\label{fig:2/3-family}The directed, semi-complete graphs $D_{2}$,
$D_{3}$, and $D_{n}$. All vertical arcs have anti-parallel counterparts.
The remaining arcs are rightwards}

\end{figure}
We proceed to compute the hat number of the digraphs of the family
$\{D_{n}\}_{n=0}^{\infty}$. First, we prove the upper bound. 
\begin{lem}
\label{lem:directed-union-upper-bound}For any digraph $D$ we have
$h\left(D\to K_{2}\right)\le\max\left\{ h\left(D\right),1/2+\left(1/4\right)h\left(D\right)\right\} $.\end{lem}
\begin{proof}
Let $\mathcal{S}$ be a strategy for $D\to K_{2}$. Denote the $K_{2}$
vertices by $x$ and $y$, and let us consider the sub-strategy played
by $x$ and $y$.
\begin{caseenv}
\item If at least one of $x$ and $y$ always tries to guess, then $\prob{\mathcal{S}}\le1/2$.
\item If at least one of $x$ and $y$ never guesses its color, without
loss of generality let it be $y$. Then, by Claims~\ref{clm:never-guess}
and~\ref{clm:drop-blind-vertices} we have $\prob{\mathcal{S}}\le h\left(D\to K_{2}-y\right)=h\left(D\to K_{1}\right)=h\left(D\right)$.
\item If both $x$ and $y$ guess their colors sometime, then each one guesses
its color with probability $1/2$ as every one of them has just one
outgoing arc. Hence, with probability at least $1/4$ at least one
is wrong. The chance of success of the strategy $\mathcal{S}$ benefits
from the behavior of the vertices of $D$ only when both $x$ and
$y$ pass, and this happens exactly with probability $1/4$ since
they see different vertices (that is, each other). Since the behavior
of the vertices of $D$ when both $x$ and $y$ pass is a strategy
$\mathcal{S}'$ for $D$, we can bound \[
\prob{\mathcal{S}}\le1/2+\left(1/4\right)\prob{\mathcal{S}'}\le1/2+\left(1/4\right)h\left(D\right).\]

\end{caseenv}
The result is established by taking $\mathcal{S}$ to be an optimal
strategy for $D\to K_{2}$.
\end{proof}
The next lemma proves the lower bound in a more general setting.
\begin{lem}
\label{lem:directed-union-lower-bound}For every positive integer
$m$ there exists $c>0$ such that for any digraph $D$ we have $h\left(D\to K_{m}\right)\ge cm/\left(m+1\right)+\left(1-c\right)h\left(D\right)$.\end{lem}
\begin{proof}
Let $\mathcal{S}$ be an optimal strategy for the digraph $D$. We
describe a strategy $\mathcal{S}'$ for the digraph $D\to K_{m}$.
Denote the vertices of $K_{m}$ by $x_{1},x_{2},\ldots,x_{m}$. 
\begin{enumerate}
\item \label{item:strategy-of-D}The vertices of $D$ pass if at most one
of $\left\{ x_{1},\ldots,x_{m}\right\} $ has a red hat; otherwise,
they behave according to the strategy $\mathcal{S}$.
\item \label{item:strategy-of-Km}For $i=1,\ldots,m$, the vertex $x_{i}$
can see the $m-1$ vertices $\left\{ x_{j}:j\neq i\right\} $. If
all of them have blue hats, then $x_{i}$ guesses red; otherwise it
passes.
\end{enumerate}
If $x_{1},\ldots,x_{m}$ all have a blue hat, then they all guess
wrong. If exactly one of them, $x_{i}$, had a red hat, then $x_{i}$
guesses correctly and all other vertices pass. All in all, conditioned
on the event $\mathcal{A}$ that at most one of $x_{1},\ldots,x_{m}$
has a red hat, the team wins with probability $m/\left(m+1\right)$.
Let $c=\prob{\mathcal{A}}=\prob{\mathrm{Bin}\left(m,1/2\right)\le2}=\left(m+1\right)2^{-m}$.
We have \[
\prob{\mathcal{S}'}=\prob{\mathcal{A}}m/\left(m+1\right)+\left(1-\prob{\mathcal{A}}\right)\prob{\mathcal{S}}=cm/\left(m+1\right)+\left(1-c\right)h\left(D\right),\]
establishing the result.\end{proof}
\begin{rem*}
In the proof of Lemma~\ref{lem:directed-union-lower-bound}, $c$
approaches zero very quickly as $m$ grows. In fact, we can have $c\ge1/2$
by using a slightly more complicated strategy. Let $C\subset\left\{ \textrm{blue},\textrm{red}\right\} ^{m}$
be a code of distance 3, and consider the packing of stars $K_{1,m}$
in the hypercube graph $H_{m}$ formed by selecting balls of radius
one around each codeword. 

The event $\mathcal{A}$ is now defined as the event that the configuration
of $x_{1},\ldots,x_{m}$ is covered by the packing. Step~\ref{item:strategy-of-D}
stays basically the same: the vertices of $D$ all pass if $\mathcal{A}$
occurred and behave according to $\mathcal{S}$ otherwise. Step~\ref{item:strategy-of-Km}
is generalized to make use of the entire packing: if $x_{i}$ sees
a configuration consistent with some codeword, it guesses the color
that disagrees with it. As before, when $\mathcal{A}$ occurs either
$m$ vertices guess wrong or exactly one guesses, and is correct.

Now the existence of codes of distance 3, length $m$, and size $\ceil{2^{m-1}/\left(m+1\right)}$
implies that $c\ge1/2$.
\end{rem*}
\bigskip{}

We use Lemmata~\ref{lem:directed-union-upper-bound} and~\ref{lem:directed-union-lower-bound}
to calculate the hat number of $D_{n}$.
\begin{cor}
For every non-negative integer $n$ we have \[
h\left(D_{n}\right)=\frac{2}{3}-\frac{1}{6}\cdot\frac{1}{4^{n}}.\]
\end{cor}
\begin{proof}
We prove the result by induction on the number $n$. For $n=0$ the
claim is obviously true as $D_{0}$ is a single isolated vertex and
$h\left(D_{0}\right)=1/2=2/3-1/6$. Let $n$ be a positive integer,
and assume that $h\left(D_{n-1}\right)=2/3-4^{1-n}/6$. Since $h\left(D_{n-1}\right)<2/3$,
by Lemma~\ref{lem:directed-union-upper-bound} we have \[
h\left(D_{n}\right)\le\max\left\{ h\left(D_{n-1}\right),1/2+\left(1/4\right)h\left(D_{n-1}\right)\right\} =1/2+\left(1/4\right)h\left(D_{n-1}\right).\]
This is matched by Lemma~\ref{lem:directed-union-lower-bound}, which
gives $h\left(D_{n}\right)\ge\left(3/4\right)\left(2/3\right)+\left(1/4\right)h\left(D_{n-1}\right)$.
Therefore \[
h\left(D_{n}\right)=\frac{1}{2}+\frac{1}{4}h\left(D_{n-1}\right)=\frac{1}{2}+\frac{1}{4}\left(\frac{2}{3}-\frac{1}{6}\cdot\frac{1}{4^{n-1}}\right)=\frac{2}{3}-\frac{1}{6}\cdot\frac{1}{4^{n}}\]
 and the result is established. 
\end{proof}
We have just proved the following.
\begin{thm}
\label{thm:2/3-minus-epsilon}For every $\varepsilon>0$ there exists
a digraph $D$ satisfying $\omega\left(D\right)=2$ such that $h\left(D\right)>2/3-\varepsilon$.
\end{thm}
Theorem~\ref{thm:2/3-minus-epsilon} can be generalized to an arbitrary
clique number $m$.
\begin{thm}
For every $\varepsilon>0$ there exists a digraph $D$ satisfying
$\omega\left(D\right)=m$ such that $h\left(D\right)>m/\left(m+1\right)-\varepsilon$.\end{thm}
\begin{proof}
Let us consider $D=K_{m}^{\to n}$, where $n=\ceil{\log_{1-c}\left(\varepsilon\right)}$
and $c$ is the appropriate constant from Lemma~\ref{lem:directed-union-lower-bound}.
By repeatedly applying the lemma we get that \[
h\left(D\right)\ge\left(1-\left(1-c\right)^{n}\right)m/\left(m+1\right)\ge\left(1-\varepsilon\right)m/\left(m+1\right)>m/\left(m+1\right)-\varepsilon,\]
as needed.
\end{proof}
\medskip{}
The natural question to ask is whether $m/\left(m+1\right)$ is the
best possible hat number of such digraphs. In the following section
we show that indeed this is the best possible, i.e., that the chance
of success $m/\left(m+1\right)$ is asymptotically optimal for digraphs
with clique number $m$.

\section{\label{sec:upper-bound}The upper bound}

Feige~\cite{Feige09+} proved that for every undirected graph $G$
we have $h\left(G\right)\le\omega\left(G\right)/\left(\omega\left(G\right)+1\right)$.
We repeat his proof, refining it a bit to show that the same holds
for digraphs.
\begin{prop}
\label{prop:hat-number-upper-bound}For every digraph $D$ we have
$h\left(D\right)\le\omega\left(D\right)/\left(\omega\left(D\right)+1\right)$.\end{prop}
\begin{proof}
Let $\mathcal{S}$ be an optimal strategy for $D$. By $W_{\mathcal{S}}$
let us denote the set of configurations in which the team wins using
the strategy $\mathcal{S}$ and by $L_{\mathcal{S}}$ let us denote
the set of configurations in which the team actively loses using the
strategy $\mathcal{S}$, that is, configurations in which $\mathcal{S}$
causes at least one wrong guess.

We define a bipartite graph $B$ whose left-hand side is $L_{\mathcal{S}}$,
and right-hand side is $W_{\mathcal{S}}$. A~losing configuration
$l\in L_{\mathcal{S}}$ is adjacent to a winning configuration $w\in W_{\mathcal{S}}$
if they differ only by one coordinate, which is the hat color of a
vertex $v\in V\left(G\right)$ that attempted to guess at these configurations.%
\footnote{Since $v$ cannot see its own hat color, it acts the same in both
hat configurations $l$ and $w$.%
} Let us examine the right and the left degrees in $B$.

\paragraph{Right degree.}

Let $w\in W_{\mathcal{S}}$ be a winning configuration, and let $v\in V\left(D\right)$
be a vertex that guesses correctly at $w$. Let $l$ be a hat configuration
identical to $w$ except in coordinate $v$. Since $v$ does not see
any difference between $w$ and $l$, it makes the same guess in $l$,
but now it is incorrect.

Therefore $l\in L_{\mathcal{S}}$ is a neighbor of $w$ in $B$, and
$d\left(w\right)\ge1$.

\paragraph{Left degree.}

Let $l\in L_{\mathcal{S}}$ be a losing configuration, and let $w_{1},\ldots,w_{d}\in W_{\mathcal{S}}$
be its neighbors in $B$, where $d=d\left(l\right)$. For every $i=1,\ldots,d$
let $v_{i}\in V\left(D\right)$ be the coordinate at which $l$ and
$w_{i}$ differ.

Assume for the sake of contradiction that some arc $v_{i}\to v_{j}$
is not present in $D$. By the definition of $v_{i}$, it makes a
correct guess at the configuration $w_{i}$. It cannot tell $w_{i}$
apart from $l$, and thus it makes the same, now wrong, guess at the
configuration $l$. But then it must make the same incorrect guess
at the configuration $w_{j}$, which only differs from $l$ by the
color of $v_{j}$, unseen by $v_{i}$. This contradicts the fact that
$w_{j}$ is a a winning configuration. 

Therefore $\left\{ v_{i}\right\} _{i=1}^{d}$ is a clique in $\skel D$
and $d=d\left(l\right)\le\omega\left(\skel D\right)=\omega\left(D\right)$.

\bigskip{}
We have shown that the right degree in $B$ is at least one and the
left degree in $B$ is at most $\omega\left(D\right)$. This implies
that $\left|W_{\mathcal{S}}\right|\le\left|E\left(B\right)\right|\le\omega\left(D\right)\left|L_{\mathcal{S}}\right|$
and consequently \[
h\left(D\right)=\prob{\mathcal{S}}=\left|W_{\mathcal{S}}\right|\cdot2^{-\left|V\left(D\right)\right|}\le\left|W_{\mathcal{S}}\right|/\left(\left|W_{\mathcal{S}}\right|+\left|L_{\mathcal{S}}\right|\right)\le\omega\left(D\right)/\left(\omega\left(D\right)+1\right),\]
establishing the result.\end{proof}
\begin{rem*}
Observe that for a digraph $D$, the hat number $h\left(D\right)$
is always a rational number whose denominator is a power of two. Thus,
$h\left(D\right)<\omega\left(D\right)/\left(\omega\left(D\right)+1\right)$
unless $\omega\left(D\right)+1$ is a power of two.%
\footnote{When $\omega\left(D\right)+1=2^{k}$ is a power of two, the upper
bound is met by a complete graph $K_{2^{k}-1}$ as $h\left(K_{2^{k}-1}\right)=\left(2^{k}-1\right)/2^{k}$.%
}\end{rem*}
\begin{cor}
For every tournament $T$ we have $h\left(T\right)=1/2$.\end{cor}
\begin{proof}
Apply Proposition~\ref{prop:hat-number-upper-bound} with $\omega\left(T\right)=1$.
The lower bound is by Claim~\ref{clm:at-least-half}.
\end{proof}
{}

\subsection*{Acknowledgments}

The authors thank Noga Alon, Uriel Feige, and Po-Shen Loh for useful
discussions and comments.

\end{document}